\documentclass[12pt]{article}

\usepackage{amsmath, amssymb, amsthm}
\usepackage{mathrsfs}
\usepackage{bm}

\usepackage[utf8]{inputenc}
\usepackage[T1]{fontenc}
\usepackage{lmodern}
\usepackage{microtype}
\usepackage{setspace}
\usepackage[hypertexnames=false]{hyperref}
\usepackage{todonotes}
\usepackage{epigraph}
\usepackage{algorithm}
\usepackage{algpseudocode}
\usepackage{graphicx}

\newtheorem{theorem}{Theorem}[section]
\newtheorem{proposition}[theorem]{Proposition}
\newtheorem{lemma}[theorem]{Lemma}
\newtheorem{corollary}[theorem]{Corollary}
\newtheorem{definition}[theorem]{Definition}

\newtheorem{remark}[theorem]{Remark}

\newcommand{\N}{\mathbb{N}}
\newcommand{\E}{\mathbb{E}}

\newcommand{\bbP}{\mathbb{P}}

\usepackage[
  backend=biber,
  style=phys,
  sorting=none,
  biblabel=brackets
]{biblatex}
\title{Multiplicative Turing Ensembles, \\ Pareto's Law, and Creativity}

\author{Alexander Kolpakov\textsuperscript{*} \\ University of Austin \\ Austin TX, USA \\ \href{akolpakov@uaustin.org}{akolpakov@uaustin.org} 
   \and Aidan Rocke \\ Solomonoff Consulting \\ Amsterdam, The Netherlands \\ \href{rockeaidan@gmail.com}{rockeaidan@gmail.com} }

\begin{document}

\maketitle

\begin{abstract}
We study integer-valued multiplicative dynamics driven by i.i.d. prime multipliers and connect their macroscopic statistics to universal codelengths. We introduce the Multiplicative Turing Ensemble (MTE) and show how it arises naturally -- though not uniquely -- from ensembles of probabilistic Turing machines. Our modeling principle is variational: taking Elias' Omega codelength as an energy and imposing maximum entropy constraints yields a canonical Gibbs prior on integers and, by restriction, on primes. Under mild tail assumptions, this prior induces exponential tails for log-multipliers (up to slowly varying corrections), which in turn generate Pareto-type tails for additive gaps, with the survival exponent shifted by summation over primes. We also prove time-average laws for the Omega codelength along MTE trajectories. Empirically, Debian, PyPI, and CRAN package-size histograms have fitted Omega slopes well below the pure-Omega value $\log 2$, indicating heavier-than-pure-Omega tails within this energy scale. Taken together, the theory--data comparison suggests a qualitative split: machine-adapted regimes (Gibbs-aligned, finite first moment) exhibit clean averaging behavior, whereas human-generated complexity appears to sit beyond this regime, with tails heavy enough to produce an unbounded first moment, and therefore no averaging of the same kind.
\end{abstract}

\newpage

\begingroup
\renewcommand{\thefootnote}{*}
\footnotetext{Corresponding author; 522 N Congress Ave STE 300, Austin, TX 78701}
\endgroup

\section{Introduction}
Multiplicative stochastic models often yield power-law statistics through renewal or Kesten-type mechanisms\cite{Kesten73,Mandelbrot74}. Here the relevant ``gaps'' are not gaps between consecutive primes, but additive jumps $X_{t+1}-X_t$ generated by multiplicative prime-valued updates. Universal integer codes, most notably Elias' $\omega$ code \cite{Elias1975}, provide codelengths that approximate Kolmogorov complexity up to logarithmic terms.

We bring these strands together via the \emph{Multiplicative Turing Ensemble} (MTE), effectively a prime-multiplier Markov chain that can be motivated—but not uniquely determined—by ensembles of probabilistic Turing machines. First, we provide a variational derivation of a natural multiplier law from an $\omega$--based Gibbs principle. Then we show that, under mild tail assumptions, the additive gaps exhibit asymptotic Pareto behavior, and we prove averaging results for $\omega$ codelength along MTE trajectories. Finally, we use Debian, PyPI, and CRAN package-size data as an Omega-tail diagnostic: whether observed codelength histograms occur at the pure-Omega slope or in a heavier scaled-Omega regime.

\paragraph{Logarithm convention.}
Unless explicitly marked as $\log_2$, logarithms are natural logarithms. Codelengths are measured in bits and therefore use $\log_2$.

\section{Model and Preliminaries}\label{sec:model}

\paragraph{Probabilistic Turing Machines.}
Fix a probabilistic Turing machine (PTM) $\Pi$ that on each discrete step emits one of three symbols $\{0,1,S\}$ with probabilities $p_0$, $p_1$, and $p_S$, respectively, such that $0 < p_0, p_1, p_S < 1$, $p_0+p_1+p_S=1$.

Symbols $0$ and $1$ are appended to the output tape; the special symbol $S$ causes the machine to halt \emph{without} being written to the tape. Thus each run of $\Pi$ produces a finite binary string $x\in\{0,1\}^*$. This PTM is a standard way to model random finite outputs with halting and induces a semimeasure on $\{0,1\}^*$, cf.\ \cite[Ch.~7]{arora2009computational}, \cite[Ch.~4]{li1997introduction}, \cite[Ch.~6-7]{calude2002random}.

The probability that a specific string $x=x_1x_2\cdots x_n\in\{0,1\}^n$ is produced and the machine halts immediately afterwards is
\begin{equation}\label{eq:ptm-string-prob}
\mathbb{P}_\Pi(x)\;=\;p_S\prod_{i=1}^n p_{x_i},\qquad\text{where }p_{x_i}:=\begin{cases}p_0,&x_i=0\\ p_1,&x_i=1.\end{cases}
\end{equation}

The output length $|x|$ has geometric distribution
$\mathbb{P}(|x|=n)=p_S(1-p_S)^n$.

Let $\mathrm{bin}:\{0,1\}^*\to\mathbb{N}$ be the base-2 evaluation map: $\mathrm{bin}(\epsilon)=0$ (empty string) and $\mathrm{bin}(x_1\cdots x_n) = \sum_{i=1}^n x_i 2^{n-i}$ for $x_i \in \{0,1\}$. Note that $\mathrm{bin}$ is \emph{not injective}: strings differing only in leading zeros map to the same integer (e.g., $\mathrm{bin}(``1'')=\mathrm{bin}(``01'')=\mathrm{bin}(``001'')=1$).

Define the ``prime filter'' event
\[
\mathsf{Prime}=\{x\in\{0,1\}^*: \mathrm{bin}(x)\text{ is prime}\}.
\]
Primality is a computable predicate, hence measurable w.r.t.\ the recursive $\sigma$–algebra on $\{0,1\}^*$ \cite[Sec.~18.1]{arora2009computational}. The induced probability distribution, for any prime $p$, is the conditional law
\begin{equation}\label{eq:mu-Pi}
\mu_\Pi(\,p\,)\;:=\;\mathbb{P}_\Pi\!\left(\mathrm{bin}(X)=p\;\middle|\;\mathsf{Prime}\right) = \frac{\sum_{x:\,\mathrm{bin}(x)=p} \mathbb{P}_\Pi(x)}{\mathbb{P}_\Pi(\mathsf{Prime})}.
\end{equation}
The sum accounts for all binary representations of $p$, with or without leading zeros. Since each string $x$ of length $n$ has probability proportional to $(1-p_S)^n$, longer representations of the same $p$ (those with more leading zeros) contribute exponentially smaller weight.

\begin{lemma}[$\mu_\Pi$ well-defined and positive]\label{lem:prime-positive}
If $p_0,p_1,p_S>0$, then $$\mathbb{P}_\Pi(\mathsf{Prime})>0,$$ hence $\mu_\Pi$ in~\eqref{eq:mu-Pi} is a well-defined probability distribution on the primes.
\end{lemma}

\begin{proof}
Because $p_0,p_1>0$, every finite binary string $x$ has $\mathbb{P}_\Pi(x)>0$ by~\eqref{eq:ptm-string-prob}. Fix any prime $p$ and let $x$ be any of its binary representations; then $\mathbb{P}_\Pi(x)>0$ and $\mathrm{bin}(x)=p$, so $\mathbb{P}_\Pi(\mathsf{Prime})\ge \mathbb{P}_\Pi(x)>0$. Since primality is decidable \cite[Thm.~18.5]{arora2009computational}, the conditioning is computable relative to $\Pi$'s semimeasure.\footnote{See also \cite[Ch.~4]{li1997introduction}, \cite[Ch.~6-7]{calude2002random} for background on PTM-induced semimeasures and computable sets.}
\end{proof}

\paragraph{Equivalent viewpoints on ensembles.}
We now consider an \emph{ensemble} of PTMs $\{\Pi_i\}_{i\in I}$ indexed by a countable set $I$, with mixture weights $w_i>0$, $\sum_{i\in I}w_i=1$. Let $\mu_{\Pi_i}$ be the prime-filtered law~\eqref{eq:mu-Pi} of $\Pi_i$. The ensemble induces the mixture
\begin{equation}\label{eq:mixture-ensemble}
\mu_{\mathrm{ens}}\;=\;\sum_{i\in I} w_i\,\mu_{\Pi_i}
\end{equation}
on the set $\mathsf{Prime}$.

There are at least three operationally equivalent ways to realize $\mu_{\mathrm{ens}}$:
\begin{enumerate}
\item[(A)] \emph{Mixture-once:} sample $I\sim w$ once, run $\Pi_I$ once, and condition on the prime event.
\item[(B)] \emph{Consecutive runs of a single PTM:} fix any $i$ and run $\Pi_i$ repeatedly, conditioning each run on the prime event; if before observation you also randomize $i\sim w$, the marginal law of a single observed prime equals \eqref{eq:mixture-ensemble}.
\item[(C)] \emph{Single PTM with latent choice:} define a new PTM $\widetilde\Pi$ that starts by sampling $I\sim w$ using its internal randomness, then simulates $\Pi_I$ (this is a standard PTM construction; cf.\ \cite[Sec.~7.2]{li1997introduction}). Condition on primality at the end. The resulting prime distribution is exactly $\mu_{\mathrm{ens}}$.
\end{enumerate}

\begin{proposition}[Equivalence of ensemble viewpoints]\label{prop:equiv}
The three procedures (A)–(C) produce the same probability distribution on primes, namely $\mu_{\mathrm{ens}}$ in~\eqref{eq:mixture-ensemble}.
\end{proposition}

\begin{proof}
\begin{enumerate}
    \item[(A)] The mixture follows readily from the law of total probability: $$\mathbb{P}(\cdot\mid\mathsf{Prime})=\sum_i w_i\,\mathbb{P}_i(\cdot\mid\mathsf{Prime}).$$ 
    \item[(B)] The first observed prime from a randomly chosen $i\sim w$ has marginal $\sum_i w_i\,\mu_{\Pi_i}$. Independence between runs follows by assuming fresh randomness each time; see \cite[Ch.~6--7]{calude2002random}.
    \item[(C)] By construction $\widetilde\Pi$ simulates the two-stage mixture in one PTM; conditioning commutes with the initial latent draw. All steps are standard for PTM mixtures and semimeasures; see \cite[Ch.~4]{li1997introduction}.
\end{enumerate}
\end{proof}

\paragraph{Towards Multiplicative Turing Ensembles.}
The PTM ensemble construction above is \emph{one} natural route to a prime distribution $\pi$, which is neither unique nor necessary. Any probability mass function $\pi=\{\pi_p\}_{p\in\mathsf{Prime}}$ on the primes suffices to define an MTE. The PTM perspective provides intuition, especially for the connection to Kolmogorov complexity and prefix codes, but the essential object is simply $\pi$ itself.

For completeness, let us consider a PTM ensemble $\{\Pi_i\}_{i\in I}$ with mixture weights $w_i>0$, $\sum_i w_i=1$, and let $\mu_i = \mu_{\Pi_i}$ denote the prime-filtered law (cf.\ Lemma~\ref{lem:prime-positive}). Then the induced ensemble prime law, for $p$ prime, is
\[
  \pi_p = \mu_{\mathrm{ens}}(p) = \sum_{i\in I} w_i\,\mu_i(p).
\]
By Proposition~\ref{prop:equiv}, $\pi$ is the marginal distribution of the prime output when we pick a PTM from the ensemble according to $w = \{w_i\}_{i\in I}$ and run it once. However, $\pi$ could equally well be postulated directly, or derived from other principles (e.g., maximum entropy, as in Section~\ref{sec:variational}).

\medskip
\noindent\textbf{Ensemble Aggregation.}
To produce a time series from successive prime outputs, we define a \emph{current state}
$X_t\in\mathbb{N}_{\ge 1}$ and the \emph{multiplicative update law}
\[
  X_{t+1}=X_t\cdot P_{t+1},\qquad P_{t+1}\sim \pi \ \text{ i.i.d.},\qquad X_0=1.
\]
Because the new multiplier $P_{t+1}$ is drawn independently of $X_t$ and takes values in $\mathsf{Prime}$, the induced process $(X_t)_{t\ge 0}$ is a \emph{time-homogeneous Markov chain} on $\mathbb{N}_{\ge 1}$ with
transition kernel
\begin{equation}\label{eq:mult-kernel}
  \mathsf{K}(n,m)=
  \begin{cases}
    \pi_{m/n}, & \text{if } m/n\in\mathsf{Prime},\\[2pt]
    0, & \text{otherwise.}
  \end{cases}
\end{equation}

We shall add the integrability assumption $\E[\log P_1]<\infty$, which is a necessary condition on the tail of $\pi$. In the PTM ensemble case, it is equivalent to $\sum_{p}\pi_p\log p<\infty$.

Altogether combined, we have the following definition, that is intentionally formulated in a way independent of PTM ensembles. 

\begin{definition}[Multiplicative Turing Ensemble]\label{def:MTE}
Let $\{\pi_p\}_{p\in\mathcal P}$ be a probability mass function on primes. The Multiplicative Turing Ensemble (MTE) is the Markov chain $(X_t)_{t\ge 0}$ on $\N_{\ge 1}$ with
\begin{equation}\label{eq:MTE}
  X_{t+1} = X_t \cdot P_{t+1}, \qquad \bbP(P_{t+1}=p)=\pi_p, \quad X_0=1.
\end{equation}
We assume $\E[\log P_1]<\infty$.
\end{definition}

\paragraph{Prefix codes and energy functions} Let us fix the binary alphabet $\Sigma = \{0, 1\}$, and let $\Sigma^*$ be its Kleene closure (which amounts to all possible binary strings in this case). A code $\mathcal{C} \subset \Sigma^*$ consists of a countable amount of binary strings called codewords. 

A code is called uniquely decodable if there is only one way to represent a string $s \in \Sigma^*$ as a sequence of concatenated codewords $c_1*c_2*\ldots*c_k$, $c_i\in\mathcal{C}$, for some $k\in \N$. Here ``$*$'' denotes string concatenation, and $k=0$ corresponds to the case of an empty representation. 

A code $\mathcal{C}$ is called complete if it is maximal within the class of uniquely decodable code: adding any $c\in \Sigma^*$, $c \notin \mathcal{C}$, results in $\mathcal{C}\cup\{c\}$ being \emph{not} uniquely decodable. 

Elias' $\omega$ code \cite{Elias1975} is an example of a uniquely decodable code on integers that is complete, and that comes from an iterative renormalization procedure \cite{KR-elias-2025}.

\begin{definition}[Elias $\omega$ codelength]\label{def:omega}
Let $n_0=n$ and recursively define $n_{j+1}=\lfloor \log_2 n_j\rfloor+1$.  
Stop at the first $t$ such that $n_t=1$.  
Then the Elias $\omega$ codelength of $n$ is
\begin{equation}\label{eq:omega-def}
      \ell_\omega(n) \;=\; 1 + \sum_{j=0}^{t-1} \big(\lfloor \log_2 n_j\rfloor + 1\big).
\end{equation}

In particular,
\begin{equation}\label{eq:omega-asymp}
  \ell_\omega(n) \;=\; \log_2 n + \log_2 \log_2 n + \Theta(\log_2 \log_2 \log_2 n), \qquad n\to\infty.
\end{equation}
\end{definition}

More precisely, the error term has the form \[\log_2 \log_2 \log_2 n + \log_2 \log_2 \log_2 \log_2 n + \cdots\] down to $O(1)$, a sum that converges to a bounded iterated logarithm. For practical purposes, $\ell_\omega(n) \approx \log_2 n + \log_2 \log_2 n$ for all but astronomically large $n$.

\begin{remark}
The Elias $\omega$--length and the usual bit-length $\lfloor \log_2 n \rfloor + 1$ differ by $\log_2 \log_2 n + \ldots $ (lower-order terms). This logarithmic overhead is the price of self-delimiting encoding: $\omega$ encodes not only $n$ but also the length of that encoding, recursively, without requiring external length markers.
\end{remark} 

\paragraph{Codelength from basic principles.}
Let us introduce an integer ``energy'' $E:\mathbb{N}\to\mathbb{R}_+$ function that we require to be 
\begin{enumerate}
    \item[i.] computable and normalized in the sense of the Kraft--McMillan \cite{Kraft1949, McMillan1956} inequality:
                \begin{equation}
                    \sum_{n=1}^\infty 2^{-E(n)}\le 1;
                \end{equation}
    \item[ii.] compressing under binary scaling, which allows for an efficient representation and reducing complexity as only binary exponents are necessary for encoding:
                \begin{equation}
                    E(2^m n)\;\le\;E(n)+m+E(m)+O(1);
                \end{equation}
    \item[iii.] tight in the sense that the bound is met up to $O(1)$ infinitely often in $m$ for each fixed $n$.
\end{enumerate}

Setting $n=1$ in axiom (ii) yields
\begin{equation}\label{eq:E-recursion}
    E(2^m) \leq m + E(m) + O(1).
\end{equation}
By axiom (iii), we have that infinitely often
\begin{equation}
    E(2^m) = m + E(m) + O(1),
\end{equation}
which is, up to $O(1)$, the Elias' $\omega$ recursion
\begin{equation}
    \ell_\omega(2^m) = 1 + m + \ell_\omega(m).
\end{equation}

We now record a precise mass bound. The result is not a pointwise uniqueness theorem. It says that, under Kraft--McMillan normalization, a fixed improvement over $\ell_\omega$ can occur only on a set of small $\omega$--mass.

\begin{proposition}[No large Gibbs-mass improvement]\label{prop:uniqueness}
Let $E:\mathbb{N}\to\mathbb{R}_+$ satisfy the Kraft--McMillan inequality in axiom (i). For every $c>0$,
\begin{equation}
    \sum_{\{n:\,E(n)\le \ell_\omega(n)-c\}}2^{-\ell_\omega(n)}\le 2^{-c}.
\end{equation}
Consequently, no prefix-free energy can beat $\ell_\omega$ by a fixed margin $c$ on any set of integers whose $\omega$--Gibbs mass is larger than $2^{-c}$.
\end{proposition}

\begin{proof}
For $c>0$, we have
\[
1 \ge \sum_{E(n)\le \ell_\omega(n)-c} 2^{-E(n)} \ge \sum_{E(n)\le \ell_\omega(n)-c} 2^{-(\ell_\omega(n)-c)} = 2^{c}\!\!\sum_{E(n)\le \ell_\omega(n)-c} 2^{-\ell_\omega(n)},
\]
which is the claimed bound.
\end{proof}

\begin{remark}
Axioms (ii)--(iii) supply the structural motivation for choosing $\ell_\omega$: they encode the self-delimiting recursion
$E(2^m)\simeq m+E(m)$, which ordinary bit-length lacks because the recursive length marker is unbounded in $m$. Proposition~\ref{prop:uniqueness} gives the complementary obstruction: after Kraft--McMillan normalization, fixed-margin compression below $\ell_\omega$ is possible only on sets of exponentially small $\omega$--mass. Thus $\ell_\omega$ is canonical up to sparse exceptional shortcuts, not pointwise unique.
\end{remark}

We shall show that the MaxEnt prior with the canonical choice $E=\ell_\omega$,
\begin{equation}
\pi_p \propto 2^{-\ell_\omega(p)}, 
\end{equation}
is computable, normalizes on primes, and is compressing.

\section{Variational characterization of the multiplier law}\label{sec:variational}

In Section~\ref{sec:model} we established that an MTE is determined by a choice of the prime law $\pi$. Here we try to determine which properties should $\pi$ satisfy for the ensemble to be \emph{algorithmically natural}. 

We adopt an information-theoretic perspective: if the MTE is to model random integer generation via computational processes, then the distribution on primes should respect the intrinsic complexity of representing integers. This leads naturally to codelength based energy functions.

\medskip
\noindent\textbf{Maximum-entropy framework.}
Given an energy function $E:\mathbb{N} \to \mathbb{R}_+$, consider the maximum-entropy problem on $\N$:
\begin{align}\label{eq:MEP}
  &H(P) = -\sum_n P(n)\log P(n) \to \max \\
  &\sum_n P(n)=1,\; \sum_n P(n)E(n)=C,
\end{align}
where $C > 0$ is a fixed expected energy.

The solution is the Gibbs law
\begin{equation}\label{eq:gibbs-law}
    P(n) = Z^{-1}\cdot 2^{-\lambda E(n)}, \text{ with }\; Z = \sum_n 2^{-\lambda E(n)},
\end{equation}
where $\lambda>0$ is the Lagrange multiplier chosen so that $\E[E]=C$. We shall use base--$2$ logarithms and exponentials throughout, since $E$ will be measured in bits.

\medskip
The framework~\eqref{eq:MEP}--\eqref{eq:gibbs-law} works for \emph{any} energy function $E$. We offer three justifications why $E = \ell_\omega$ is a reasonable choice:

\begin{enumerate}
\item \emph{Universality:} Elias' $\omega$ code is a universal prefix-free code for the integers whose codelength obeys $\ell_\omega(n) = \log_2 n + O(\log_2\log_2 n)$ \cite{Elias1975}. By the incompressibility method, for almost all integers $n$ we have $K(n) = m + O(1) = \log_2 n + O(1)$ \cite[Ch.~2]{li1997introduction}. Hence, for almost all $n$,
    \[                                                              
    \ell_\omega(n) = K(n) + O(\log K(n)),
    \]
i.e.\ $\ell_\omega$ matches prefix Kolmogorov complexity up to a logarithmic additive term while remaining fully computable \cite{Elias1975, li1997introduction, shen2015}. Using $\ell_\omega$ as energy makes the Gibbs prior computable while capturing algorithmic complexity.

\item \emph{Self-delimiting structure:} As shown in Section~\ref{sec:model}, $\ell_\omega$ satisfies the self-referential recursion $E(2^m) = m + E(m) + O(1)$, encoding not only the data but also the data's length. This self-delimiting property is essential for prefix-free codes and ensures $\sum_n 2^{-\ell_\omega(n)} \leq 1$, which is the Kraft--McMillan inequality \cite{Kraft1949, McMillan1956}. 

\item \emph{Operational meaning:} Restricting to primes, the prior $\pi_p \propto 2^{-\ell_\omega(p)}$ assigns probability inversely proportional to the codelength needed to specify $p$. This is the natural measure on primes induced by a fair coin-flipping process that generates random bitstrings and filters for primality, as in the PTM construction of Section~\ref{sec:model}.
\end{enumerate}

Proposition~\ref{prop:uniqueness} shows that, under Kraft--McMillan normalization, any fixed-margin improvement over $\ell_\omega$ can occur only on a set of small $\omega$--Gibbs mass. Together with the self-delimiting scaling motivation in axioms (ii)--(iii), this identifies $\ell_\omega$ as the canonical computable energy for the variational model. The usual bit-length fails the tightness requirement in axiom (iii), because the recursive length overhead $E(m)$ is unbounded.

The empirical evaluation in Section~\ref{sec:empirics} tests the finite-support codelength family $q_a(\ell)\propto \exp(-a\ell)$, where the pure Omega slope is $a=\log 2$ and smaller $a$ corresponds to a heavier Omega tail.

\begin{lemma}[Normalisation over primes]\label{lemma:normalization}
For the Elias' $\omega$--length, we have
\[
\sum_{p\in \mathsf{Prime}} 2^{-\ell_\omega(p)} < 1.
\]
\end{lemma}
\begin{proof}
    The codelength function $\ell_\omega(n)$ satisfies the Kraft--McMillan inequality \cite{Kraft1949, McMillan1956}, as Elias' $\omega$--code is prefix-free by construction. Then
    \[
        \sum_{p\in \mathsf{Prime}} 2^{-\ell_\omega(p)} < \sum_{n\in \N} 2^{-\ell_\omega(n)} \leq 1. 
    \]
\end{proof}

The other properties of $E(n) = \ell_\omega(n)$, such as being compressing and non-degenerate, follow from the recursion $\ell_\omega(2^m n) = \ell_\omega(n) + m + \ell_\omega(m) + O(1)$, which is verified by direct calculation from Definition~\ref{def:omega}.

\paragraph{The pure Gibbs prior and its tail behavior.}
With $\lambda=1$ in \eqref{eq:gibbs-law} (i.e., setting the Lagrange multiplier equal to $\log 2$), the pure $\omega$ prior on integers is
\begin{equation}\label{eq:pure-omega-prior}
m_\omega(n) = 2^{-\ell_\omega(n)},
\end{equation}
since $\sum_{k} 2^{-\ell_\omega(k)} = 1$ by \cite{KR-elias-2025}. 
Restricted to primes, this gives $\pi_p^{\mathrm{pure}} \propto 2^{-\ell_\omega(p)}$.

Using the asymptotic \eqref{eq:omega-asymp}, $$\ell_\omega(p) = \log_2 p + \log_2 \log_2 p + \Theta(\log_2 \log_2 \log_2 p)$$, we have
\begin{equation}\label{eq:omega-tail}
2^{-\ell_\omega(p)} = p^{-1} (\log_2 p)^{-1} (\log_2 \log_2 p)^{-\Theta(1)}.
\end{equation}
Converting to natural logarithm, $\log_2 p = \log p / \log 2$, and noting that $(\log_2 \log_2 p)^{-\Theta(1)} = (\log \log p)^{-\Theta(1)}$ is a slowly varying factor, we have
\begin{equation}\label{eq:omega-regular-variation}
\pi_p^{\mathrm{pure}} \sim \frac{C}{p \log p \cdot (\log \log p)^{\Theta(1)}}\quad\text{as }p\to\infty,
\end{equation}
where $C$ is a normalization constant. The additional slowly varying factor $(\log \log p)^{-\Theta(1)}$ does not affect the regular variation index but introduces a polylogarithmic correction.

This is regularly varying with index $\lambda=1$ and slowly varying factor $L(p)=(\log p)^{-1}(\log \log p)^{-\Theta(1)}$. However, $\lambda=1$ is a \emph{boundary case}:
\begin{enumerate}
\item The moment $\mathbb{E}[\log P] = \sum_p \pi_p^{\mathrm{pure}} \log p \sim \sum_p \frac{1}{p \cdot (\log \log p)^{\Theta(1)}} = \infty$, so the integrability assumption of Definition~\ref{def:MTE} fails.
\item Theorems~\ref{th:conditional} and~\ref{th:mixture-tail} require $\lambda>1$ for their proofs to apply.
\end{enumerate}

\paragraph{The scaled Gibbs prior.}
To obtain an MTE with finite moments and Pareto gap tails, we use the \emph{scaled} $\omega$ prior:
\begin{equation}\label{eq:scaled-omega-prior}
\pi_p^{\mathrm{scaled}} \propto 2^{-\beta \ell_\omega(p)},\qquad \beta > 1.
\end{equation}
Then
\begin{equation}\label{eq:scaled-tail}
\pi_p^{\mathrm{scaled}} \sim C_\beta \, p^{-\beta} (\log p)^{-\beta} (\log \log p)^{-\Theta(1)},
\end{equation}
which is regularly varying with index $\lambda=\beta>1$ and slowly varying factor $L(p)=(\log p)^{-\beta}(\log \log p)^{-\Theta(1)}$. For $\beta>1$:
\begin{itemize}
\item $\mathbb{E}[\log P] = \sum_p \pi_p^{\mathrm{scaled}} \log p \propto \sum_p p^{-\beta}(\log p)^{1-\beta} (\log \log p)^{-\Theta(1)} < \infty$ (by integral test).
\item Theorems~\ref{th:conditional} and~\ref{th:mixture-tail} apply, yielding Pareto-type gap survival tails with exponent $\beta-1$.
\end{itemize}

In the sequel, we adopt the scaled prior \eqref{eq:scaled-omega-prior} with $\beta>1$ as the canonical choice for MTE analysis. The pure prior $\beta=1$ serves as a limiting case but lacks the integrability properties needed for our main results.

\subsection{Gibbs alignment and consequences}\label{subsec:gibbs-alignment}

Write $m_\omega(n):=2^{-\ell_\omega(n)}$ for the Elias $\omega$ prior on $\mathbb{N}$. For any probability distribution $\mu$ on $\mathbb{N}$, define the \emph{Gibbs alignment index} by
\[
\mathcal{G}(\mu)\ :=\ D\!\left(\mu\,\Vert\,m_\omega\right)\ \in [0,\infty].
\]
By the cross-entropy identity,
\begin{equation}\label{eq:omega-cross-entropy}
\mathbb{E}_\mu[\ell_\omega]\ =\ H(\mu)\ +\ \mathcal{G}(\mu),
\end{equation}
so $\mathcal{G}(\mu)$ measures the excess of the average $\omega$ code length over entropy.

\begin{definition}\label{def:gibbs-aligned}
We call a probability distribution $\mu$ on $\mathbb{N}$ $\omega$–\emph{aligned} if $\mu\le C\,m_\omega$ pointwise. This implies $\mathcal{G}(\mu)\le \log C$.
\end{definition}

If $\mu$ is a computable probability distribution, then we have \cite[Ch.~8]{li1997introduction} that 
\[
    \E_\mu [K] = H(\mu) + O(1),
\]
where $K(n)$ is the (uncomputable) Kolmogorov complexity of $n$. 

Thus we finally obtain
\[
    \E_\mu[\ell_\omega] = \E_\mu[K] + \mathcal{G}(\mu) + O(1) = \E_\mu[K] + O(1),
\]
once $\mu$ is $\omega$--aligned.

\paragraph{Designing ensembles to be Gibbs.}
There are two natural ways to ensure that Gibbs alignment takes place:

\begin{enumerate}
\item[(D1)] \emph{Reweighting}: choose ensemble weights $w_i$ to minimize $\mathcal{G}(\mu_{\mathrm{ens}})$ under constraints (e.g., moment constraints on primes). This is equivalent to a maximum-entropy fit with energy $\ell_\omega$.
\item[(D2)] \emph{Mechanistic constraint}: require each component PTM to realize integers via a self-delimiting description whose length is $\le \ell_\omega(n)+O(1)$; then the mixture inherits $\mu\le C\,m_\omega$.
\end{enumerate}

\paragraph{Interpretation.} It seems that the most ``natural'' ensembles are those that are \emph{as Gibbs as necessary}: their prime law is within a constant factor of $2^{-\ell_\omega}$. Exactly in this regime the uncomputability of $K$ ``averages out'', MTE inherits clean averaging properties (almost-sure convergence of time averages), and the variational principle with energy $E(n) = \ell_\omega(n)$ becomes both descriptive and prescriptive.

\section{Tail structure for multipliers and gaps}\label{sec:gap}
Write $G_t:=X_{t+1}-X_t$ for the additive gap and $L_t:=\log X_{t+1}-\log X_t=\log P_{t+1}$ for the log-gap. Conditioning on $X_t=x$ gives
\begin{equation}\label{eq:GtoL}
  \bbP(G_t>u\mid X_t=x)=\bbP\big(L_t>\log(1+u/x)\big)=\bbP\big(P_{t+1}>1+u/x\big).
\end{equation}

We now assume a general tail condition on the prime law $\pi$: suppose there exist $\lambda>1$ and a slowly varying function $L$ such that
\begin{equation}\label{eq:P-regular-var}
  \pi_p\sim p^{-\lambda}\,L(p)\qquad\text{as }p\to\infty.
\end{equation}
The slowly varying factor $L$ allows for logarithmic corrections. As shown in Section~\ref{sec:variational}, the pure $\omega$ prior $\pi_p^{\mathrm{pure}}\propto 2^{-\ell_\omega(p)}$ yields $\lambda=1$ with slowly varying factor $L(p)=(\log p)^{-1}(\log \log p)^{-\Theta(1)}$, which is a boundary case, while the scaled prior $\pi_p^{\mathrm{scaled}}\propto 2^{-\beta\ell_\omega(p)}$ with $\beta>1$ yields $\lambda=\beta$ with $L(p)=(\log p)^{-\beta}(\log \log p)^{-\Theta(1)}$.

\begin{theorem}[Conditional gap tail]\label{th:conditional}
Under~\eqref{eq:P-regular-var} with $\lambda>1$, for each fixed $x>0$ and $u\to\infty$,
\begin{equation}\label{eq:cond-pareto}
  \bbP(G_t>u\mid X_t=x) \sim C_{\lambda}\, (1+u/x)^{1-\lambda}\,\widetilde{L}\big(1+u/x\big),
\end{equation}
where $\widetilde{L}(y):=L(y)/\log y$ is slowly varying and $C_{\lambda}=1/(\lambda-1)$. Thus, under the mass-tail convention~\eqref{eq:P-regular-var}, the survival exponent of the additive gap is $\lambda-1$.
\end{theorem}
\begin{proof}
From~\eqref{eq:GtoL} and~\eqref{eq:P-regular-var},
\[
\bbP(G_t>u\mid X_t=x)=\sum_{p>y}\pi_p,\qquad y:=1+u/x.
\]
By assumption~\eqref{eq:P-regular-var}, $\pi_p \sim p^{-\lambda}L(p)$ as $p \to \infty$. Standard Tauberian estimates for regularly varying prime sums, equivalently Abel summation together with the prime number theorem, give
\[
\sum_{p>y} \pi_p \sim \sum_{p>y} p^{-\lambda}L(p) \sim \int_y^\infty t^{-\lambda}L(t)\frac{dt}{\log t}.
\]
Since $L$ is slowly varying and $\lambda > 1$, the function $f(t) := t^{-\lambda}L(t)/\log t$ is regularly varying with index $-\lambda < -1$. By Karamata's Tauberian theorem \cite[Theorem 1.5.11]{BGT1987}, for regularly varying $f$ with index $-\alpha < -1$,
\[
\int_y^\infty f(t)\,dt \sim \frac{y f(y)}{\alpha-1}.
\]
Applying this with $\alpha = \lambda$ and $f(t) = t^{-\lambda}L(t)/\log t$, we obtain
\[
\int_y^\infty t^{-\lambda}\frac{L(t)}{\log t}\,dt \sim \frac{y^{1-\lambda}L(y)/\log y}{(\lambda-1)}.
\]
Setting $y = 1+u/x$ and defining $\widetilde{L}(y) := L(y)/\log y$ (which is slowly varying since $\log y$ is slowly varying), we have
\[
\bbP(G_t>u\mid X_t=x) \sim C_{\lambda} (1+u/x)^{1-\lambda} \widetilde{L}\big(1+u/x\big),
\]
where $C_{\lambda}=1/(\lambda-1)$. For $u \gg x$, $(1+u/x)^{1-\lambda} \sim (u/x)^{1-\lambda}=x^{\lambda-1}u^{1-\lambda}$, giving a Pareto survival tail of order $u^{1-\lambda}$ up to slowly varying corrections.
\end{proof}

\begin{theorem}[Unconditional mixture tail]\label{th:mixture-tail}
Let $\nu$ be any probability measure on $(0,\infty)$ with finite $(\lambda-1+\epsilon)$-moment for some $\epsilon>0$. Then, as $u\to\infty$,
\begin{equation}\label{eq:mixed-tail}
  \bbP_\nu(G>u):=\int \bbP(G>u\mid X=x)\,\nu(dx) \sim C_\nu\, u^{1-\lambda}\,\widetilde{L}(u),
\end{equation}
where $C_\nu=(\lambda-1)^{-1}\int x^{\lambda-1}\,\nu(dx)$ and $\widetilde{L}$ is as in Theorem~\ref{th:conditional}.
\end{theorem}
\begin{proof}
Let $r=\lambda-1$ and write
\[
\overline F(z):=\bbP(P-1>z)=\bbP(P>1+z).
\]
By Theorem~\ref{th:conditional}, $\overline F$ is regularly varying with index $-r$ and
\[
\overline F(u)\sim \frac{u^{-r}\widetilde L(u)}{\lambda-1}.
\]
Since $X$ and the next multiplier $P$ are independent,
\[
\bbP_\nu(G>u)=\int \overline F(u/x)\,\nu(dx).
\]
The usual Breiman product-tail argument now applies. For completeness, choose $\delta\in(r/(r+\epsilon),1)$. On $x\le u^\delta$, Potter's bounds for regularly varying functions \cite{BGT1987} give an integrable domination by a constant multiple of $x^{r+\epsilon/2}$, while for each fixed $x$,
\[
\frac{\overline F(u/x)}{\overline F(u)}\to x^r.
\]
Dominated convergence therefore gives the contribution $\overline F(u)\int x^r\,\nu(dx)(1+o(1))$ from $x\le u^\delta$. On $x>u^\delta$, the contribution is at most $\nu(x>u^\delta)=o(\overline F(u))$ by Markov's inequality, the finite $(r+\epsilon)$-moment assumption, and the choice of $\delta$. Hence
\[
\bbP_\nu(G>u) \sim \frac{u^{1-\lambda}\widetilde{L}(u)}{\lambda-1} \int x^{\lambda-1}\,\nu(dx).
\]
\end{proof}

\paragraph{MTE is transient.} Since $\E[\log P]>0$ for primes $P\ge2$, $(X_t)$ drifts to $\infty$ and admits no finite invariant measure. The unconditional tail result~\eqref{eq:mixed-tail} requires only independence of $X$ and $P_{t+1}$, not stationarity.

\section{Convergence along MTE trajectories}\label{sec:averaging}
Below we show that time averages of the Elias' $\omega$ codelength converge almost surely along MTE trajectories. The key observation is that $\ell_\omega$ is \emph{approximately additive} under multiplication, with controlled error.

\begin{lemma}[Near–additivity of $\ell_\omega$]\label{lem:omega-nearadd}
For all integers $a,b\ge 2$,
\[
\ell_\omega(ab)\;=\;\ell_\omega(a)\;+\;\ell_\omega(b)\;+\;O\!\big(\log_2\log_2(ab)\big).
\]
Moreover, the $O(\log\log)$ scale is optimal (up to constants).
\end{lemma}

\begin{proof}
We use the asymptotic formula
\begin{equation}\label{eq:omega-asymp-precise}
\ell_\omega(n)\;=\;\log_2 n\;+\;\log_2\log_2 n\;+\;\Theta\!\big(\log_2\log_2\log_2 n\big),
\end{equation}
for $n\to \infty$, from Definition~\ref{def:omega}. 

Applying \eqref{eq:omega-asymp-precise} to $ab$, $a$, and $b$ gives
\begin{align*}
&\ell_\omega(ab)
= \log_2(ab) + \log_2\log_2(ab) + \Theta\!\big(\log_2\log_2\log_2(ab)\big),\\
&\ell_\omega(a)+\ell_\omega(b)
= \log_2 a + \log_2 b + \log_2\log_2 a + \log_2\log_2 b \\
&\quad\ + \Theta\!\big(\log_2\log_2\log_2 a + \log_2\log_2\log_2 b\big).
\end{align*}
After subtracting, we get
\begin{align}
&\ell_\omega(ab)-\ell_\omega(a)-\ell_\omega(b)
= \underbrace{\log_2\log_2(ab) - \log_2\log_2 a - \log_2\log_2 b}_{\mathrm{(I)}} \nonumber\\
&\quad +\ \underbrace{\Theta\!\big(\log_2\log_2\log_2(ab)\big) - \Theta\!\big(\log_2\log_2\log_2 a + \log_2\log_2\log_2 b\big)}_{\mathrm{(II)}}.
\label{eq:diff-decomp}
\end{align}

\textbf{Term (I).} Let \(A =\log_2 a\), \(B =\log_2 b\), so that \(A,B\ge1\). Then
\[
\mathrm{(I)}=\log_2\big(A+B\big)-\log_2 A-\log_2 B
=\log_2\!\left(\frac{1}{A}+\frac{1}{B}\right).
\]
Hence
\begin{align*}
    \mathrm{(I)}\;\le\;\log_2\!\left(\frac{2}{\min\{A,B\}}\right)
\;&=\;O\!\big(1+\log_2\log_2^{-1}\min\{a,b\}\big)\\
\;&=\;O\!\big(\log_2\log_2(\max\{a,b\})\big).
\end{align*}
Trivially \(\log_2\log_2(\max\{a,b\})\le \log_2\log_2(ab)\), for $a,b \geq 2$, so
\[
\mathrm{(I)}\;=\;O\!\big(\log_2\log_2(ab)\big).
\]

\textbf{Term (II).} Since \(\log_2\log_2\log_2(\cdot)\) is increasing for \(n\) large, and
\[
\log_2\log_2\log_2 a\;+\;\log_2\log_2\log_2 b\;\le\;2\,\log_2\log_2\log_2(ab),
\]
the \(\Theta(\cdot)\) terms in \eqref{eq:diff-decomp} are bounded in magnitude by a constant multiple of \(\log_2\log_2\log_2(ab)\). In particular,
\[
\mathrm{(II)}\;=\;O\!\big(\log_2\log_2\log_2(ab)\big).
\]

Combining the bounds for (I) and (II) in \eqref{eq:diff-decomp} yields
\[
\ell_\omega(ab)-\ell_\omega(a)-\ell_\omega(b)
\;=\;O\!\big(\log_2\log_2(ab)\big),
\]
as claimed.

\textbf{Sharpness.} Let us set \(a=b\to\infty\). Then (II) is \(O(\log\log\log a)\), while (I) equals \(1-\log_2\log_2 a\) with magnitude \(\log_2\log_2 a\). Thus the \(\log\log\) scale cannot be improved in general.
\end{proof}

\begin{theorem}[Averaging along trajectories]\label{th:averaging}
If the first moment conditions $\E[\log P_1]<\infty$ and $\E[\ell_\omega(P_1)]<\infty$ are satisfied, then almost surely
\begin{equation}\label{eq:averaging}
  \lim_{t\to\infty}\frac{\ell_\omega(X_t)}{t} = \E[\log_2 P_1],
\end{equation}
and
\begin{equation}\label{eq:averaging-sum}
  \lim_{t\to\infty}\frac{1}{t}\sum_{i=1}^t \ell_\omega(P_i) = \E[\ell_\omega(P_1)].
\end{equation}
\end{theorem}
\begin{proof}
We use the explicit asymptotic for $\ell_\omega$. For any integer $n \geq 2$,
\[
\ell_\omega(n) = \log_2 n + \log_2 \log_2 n + \Theta(\log_2 \log_2 \log_2 n),
\]
from \eqref{eq:omega-asymp}. Therefore,
\begin{align*}
\ell_\omega(X_t) &= \log_2 X_t + \log_2 \log_2 X_t + \Theta(\log_2 \log_2 \log_2 X_t)\\
&= \log_2 \prod_{i=1}^t P_i + \log_2 \log_2 \prod_{i=1}^t P_i + \Theta(\log_2 \log_2 \log_2 X_t)\\
&= \sum_{i=1}^t \log_2 P_i + \log_2 \left(\sum_{i=1}^t \log_2 P_i\right) + \Theta(\log_2 \log_2 \log_2 X_t).
\end{align*}
Dividing by $t$:
\[
\frac{\ell_\omega(X_t)}{t} = \frac{1}{t}\sum_{i=1}^t \log_2 P_i + \frac{1}{t}\log_2 \left(\sum_{i=1}^t \log_2 P_i\right) + O\left(\frac{\log_2 \log_2 \log_2 X_t}{t}\right).
\]

By the Strong Law of Large Numbers, 
$$\frac{1}{t}\sum_{i=1}^t \log_2 P_i \to \E[\log_2 P_1]$$ 
almost surely.

Also, $\log_2 X_t = \sum_{i=1}^t \log_2 P_i \sim t \E[\log_2 P_1]$, so $$\log_2 \log_2 X_t = \log_2\!\left(\sum_{i=1}^t \log_2 P_i\right) = \log_2 t + \log_2 \E[\log_2 P_1] + o(1),$$ giving
\[
\frac{1}{t}\log_2 \left(\sum_{i=1}^t \log_2 P_i\right) = \frac{\log_2 t}{t} + O(1/t) \to 0,
\]
and for the error term 
$$O\left(\frac{\log_2 \log_2 \log_2 X_t}{t}\right) \to 0,$$ 
as well. Thus 
$$\frac{\ell_\omega(X_t)}{t} \to \E[\log_2 P_1]$$ 
almost surely.

By the strong law of large numbers applied to the i.i.d.\ sequence $(\ell_\omega(P_i))$ with finite mean $\E[\ell_\omega(P_1)]$, we also have
\[
\frac{1}{t}\sum_{i=1}^t \ell_\omega(P_i) \to \E[\ell_\omega(P_1)]
\]
almost surely. 
\end{proof}

\begin{corollary}[Exponential growth of $X_t$]\label{cor:exponential-growth}
If $\E[\log P_1]<\infty$, then almost surely
\[
  \lim_{t\to\infty}\frac{\log X_t}{t} = \E[\log P_1],
\]
i.e., $X_t \propto \exp(t\,\E[\log P_1])$ almost surely, as $t \to \infty$.
\end{corollary}
\begin{proof}
Since $\log X_t = \sum_{i=1}^t \log P_i$ exactly, with no error term, this is immediate from the SLLN. Thus, $\log X_t = t\, \E[\log P_1] + o(t)$, almost surely as $t \to \infty$. Then the claim follows.
\end{proof}

\paragraph{Interpretation.} Theorem~\ref{th:averaging} establishes almost sure convergence of time averages along MTE trajectories. The code length $\ell_\omega(X_t)$ of the product grows at rate $\mathbb{E}[\log_2 P_1]$ (the logarithmic average of multipliers), while the sum $\sum_{i=1}^t \ell_\omega(P_i)$ of individual code lengths grows at the slightly faster rate $\mathbb{E}[\ell_\omega(P_1)]$. The difference arises from the logarithmic overhead in $\ell_\omega$.

Note that this is \emph{not} an ergodic theorem: the Markov chain $(X_t)$ is transient (drifts to $\infty$) and admits no stationary distribution. Instead, the results follow from the strong law of large numbers applied to the i.i.d.\ sequence $(P_t)$, combined with the explicit asymptotic \eqref{eq:omega-asymp}.

\section{Empirical evaluation}\label{sec:empirics}

Below we use three software-size datasets as diagnostics for Omega tail regimes. We do not perform general model selection over software-size laws. The test is restricted to the Omega codelength family: after converting byte sizes into Elias Omega codelengths, do the histograms follow the pure-Omega slope $a=\log 2$, or do they require a smaller slope, corresponding to heavier tails within this energy scale?

\paragraph{Datasets.} (i) Debian \texttt{stable/main/binary-amd64} package archive sizes (68,755 packages); (ii) PyPI latest-release file sizes for the top 750 most-downloaded projects (8,797 files); (iii) a deterministic sample of 750 CRAN source package archive sizes from the CRAN package index.

\paragraph{Protocol.} For each dataset, we compute $\ell_\omega(n)$ for every file (package) size $n$, then form the empirical histogram $P_{\mathrm{obs}}(\ell)$ over codelength values $\ell$. Specifically, if $N_\ell^{\mathrm{obs}}$ is the number of observed sizes with $\ell_\omega(n)=\ell$, then
\[
P_{\mathrm{obs}}(\ell) = \frac{N_\ell^{\mathrm{obs}}}{\sum_{\ell'} N_{\ell'}^{\mathrm{obs}}}.
\]

We fit only the one-parameter Omega codelength family
\[
q_a(\ell)=\frac{\exp(-a\ell)}{\sum_{\ell'\in S}\exp(-a\ell')},
\]
where $S$ is the observed codelength support after dropping a small number of the smallest codelength bins. The normalizing constant is determined by $S$; it is not an additional free intercept. The pure Omega diagnostic is the fixed value $a=\log 2$. The fitted scaled-Omega slope is the multinomial maximum-likelihood estimate, equivalently the exponential-family moment match for $\ell$. Bootstrap intervals below use 1000 multinomial resamples.

The logarithmic correction $L$ in the asymptotic theory is not estimated nonparametrically in these finite samples; the empirical diagnostic uses exact Omega codelengths and finite-support normalization. We write KL for the Kullback--Leibler divergence $D_{\rm KL}(P_{\rm obs}\Vert q)$. The fitted slopes and divergences are summarized in Table~\ref{tab:omega-fits}, with the data depicted in Figures~\ref{fig:debian-omega}, \ref{fig:pypi-omega}, and \ref{fig:cran-omega}, correspondingly. 

\begin{table}[htbp]
\centering
\small
\begin{tabular}{lcccc}
\hline
Dataset & $n$ & $\widehat a$ (95\% CI) & KL$_{\rm pure}$ & KL$_{\rm scaled}$ \\
\hline
Debian .deb sizes & 52,399 & 0.1979 [0.1965, 0.1993] & 1.6066 & 0.0998 \\
PyPI release files & 5,093 & 0.3047 [0.2963, 0.3135] & 0.4226 & 0.0734 \\
CRAN source archives & 472 & 0.0734 [0.0339, 0.1168] & 0.6598 & 0.0330 \\
\hline
\end{tabular}
\caption{Empirical Omega-tail fits for software package-size datasets. The pure Omega reference has slope $\log 2$, while the scaled-Omega law fits the finite-support slope $\widehat a$.}
\label{tab:omega-fits}
\end{table}

\begin{figure}[htbp]
\centering
\includegraphics[width=0.47\textwidth]{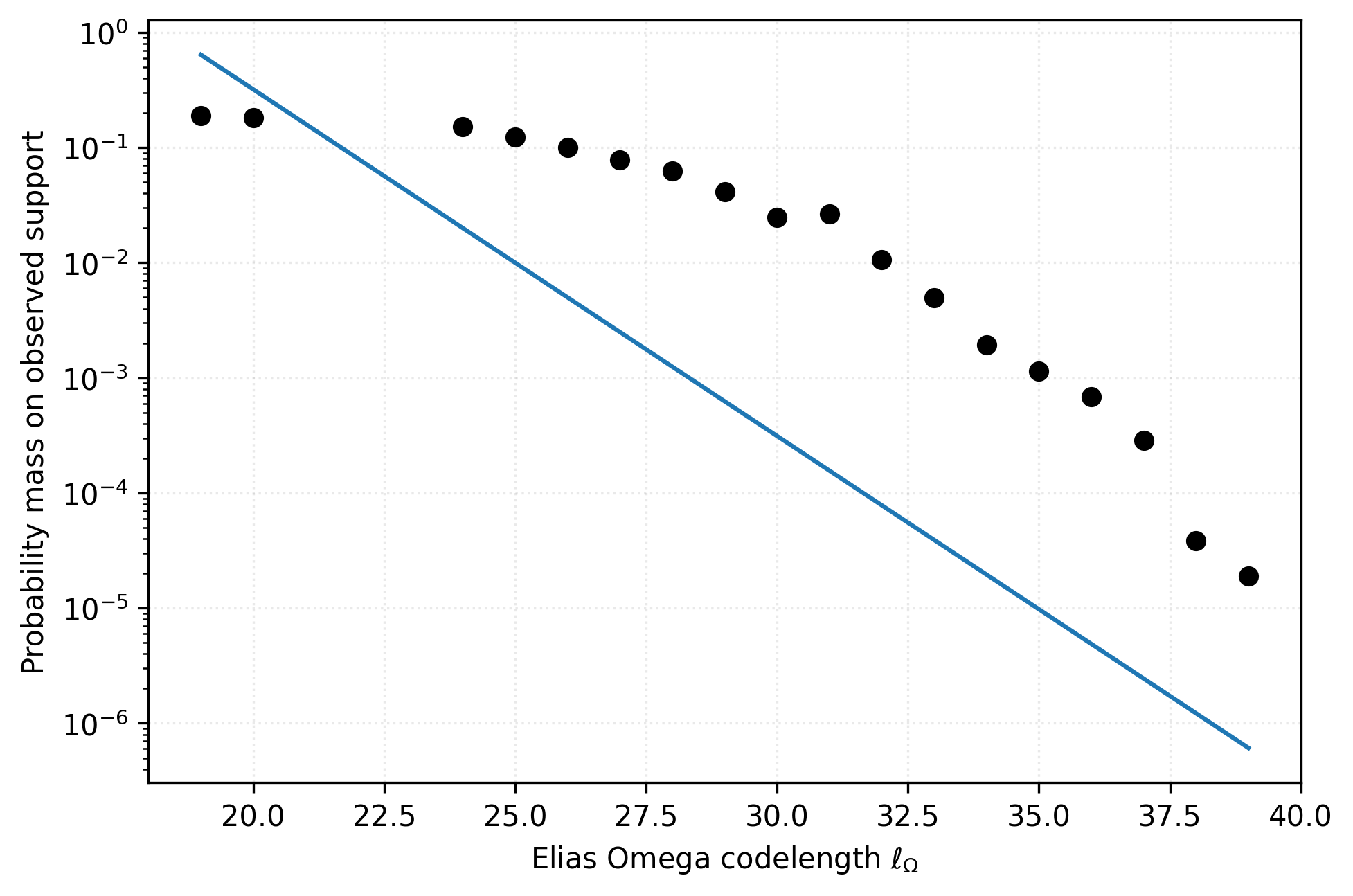}\hfill
\includegraphics[width=0.47\textwidth]{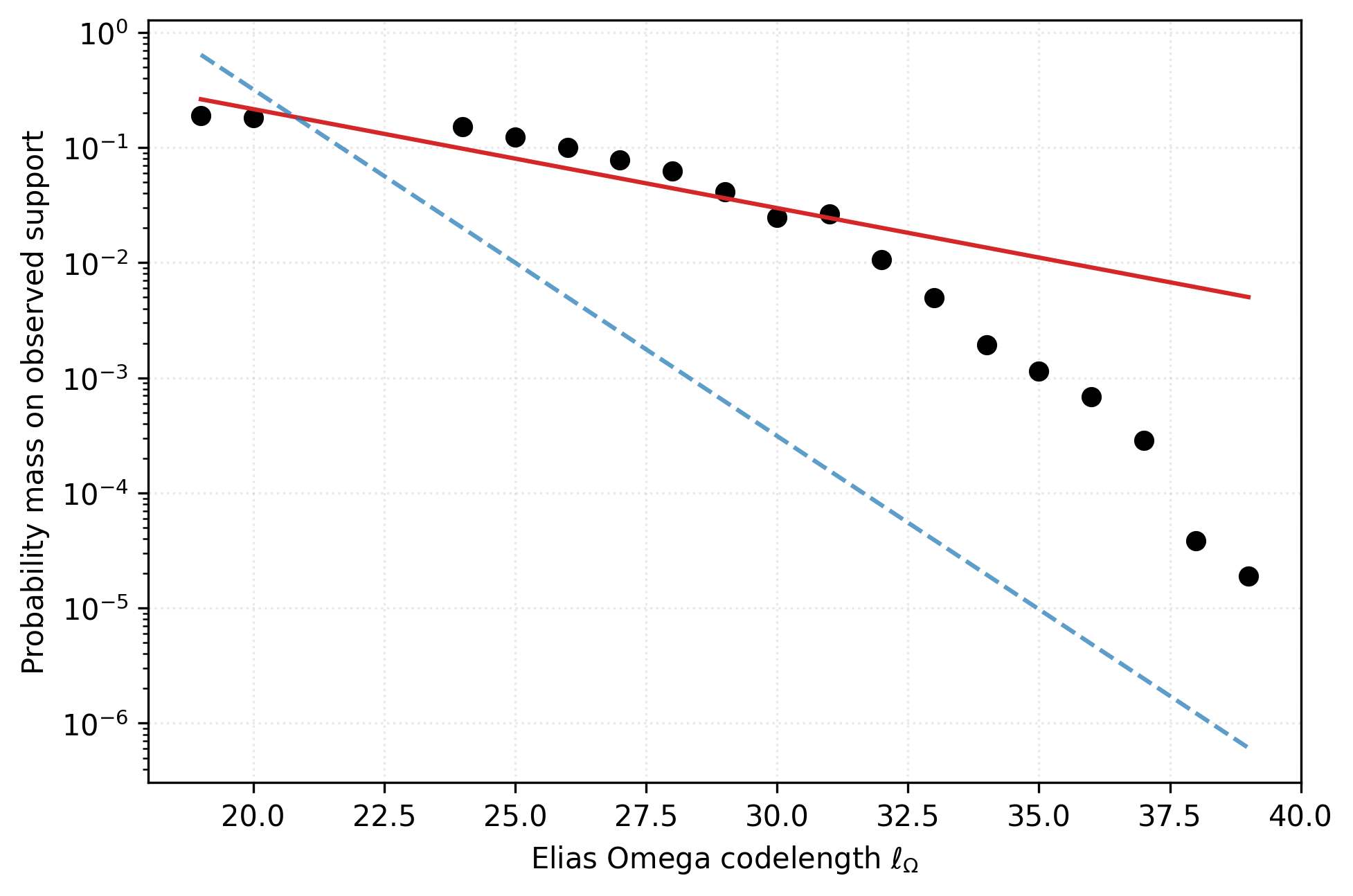}
\caption{Debian binary package archive sizes. Black points are empirical codelength frequencies. Left: blue line is the pure Omega reference. Right: red line is the fitted scaled-Omega law and dashed blue line is the pure Omega reference.}
\label{fig:debian-omega}
\end{figure}

\begin{figure}[htbp]
\centering
\includegraphics[width=0.47\textwidth]{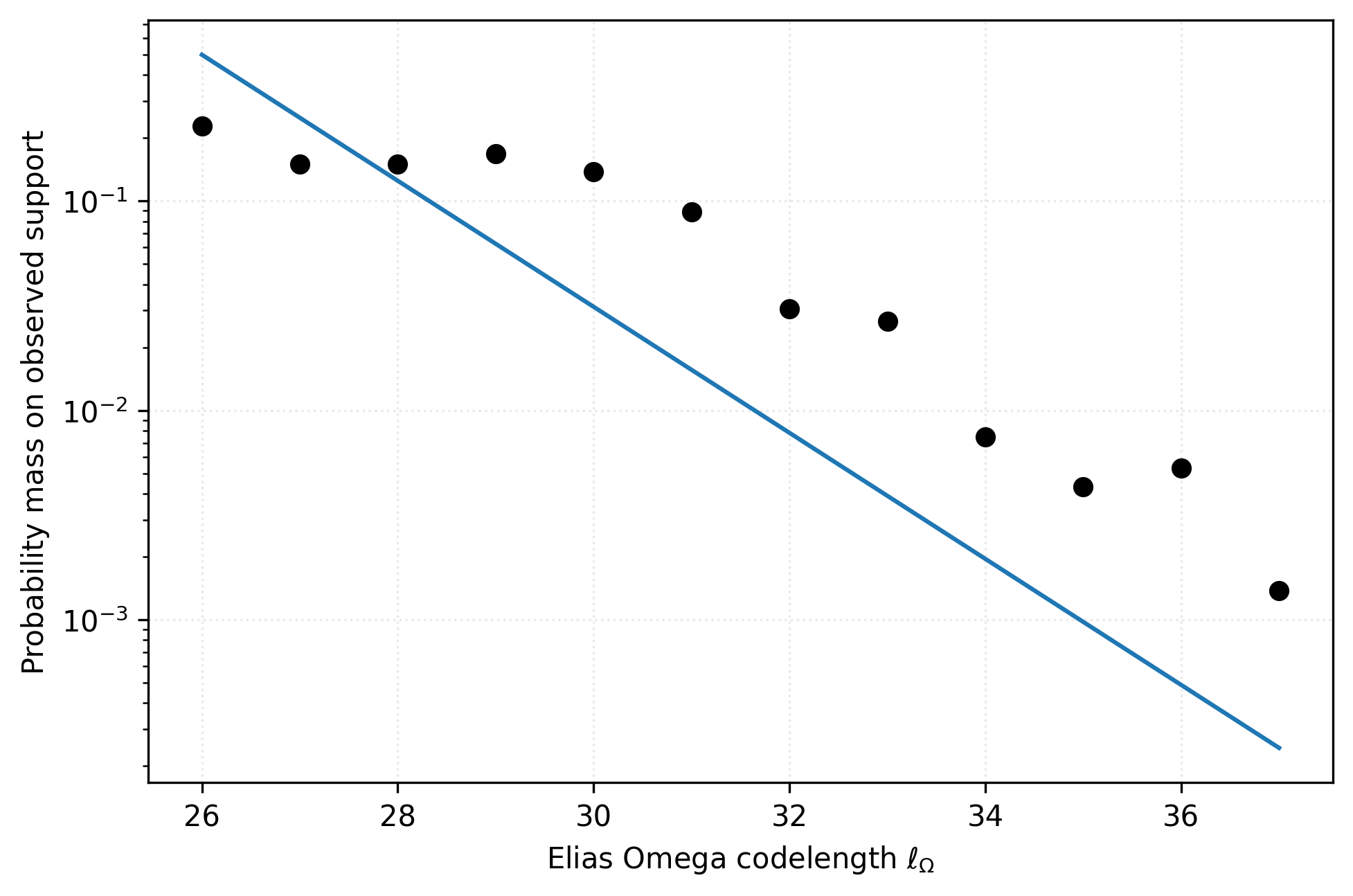}\hfill
\includegraphics[width=0.47\textwidth]{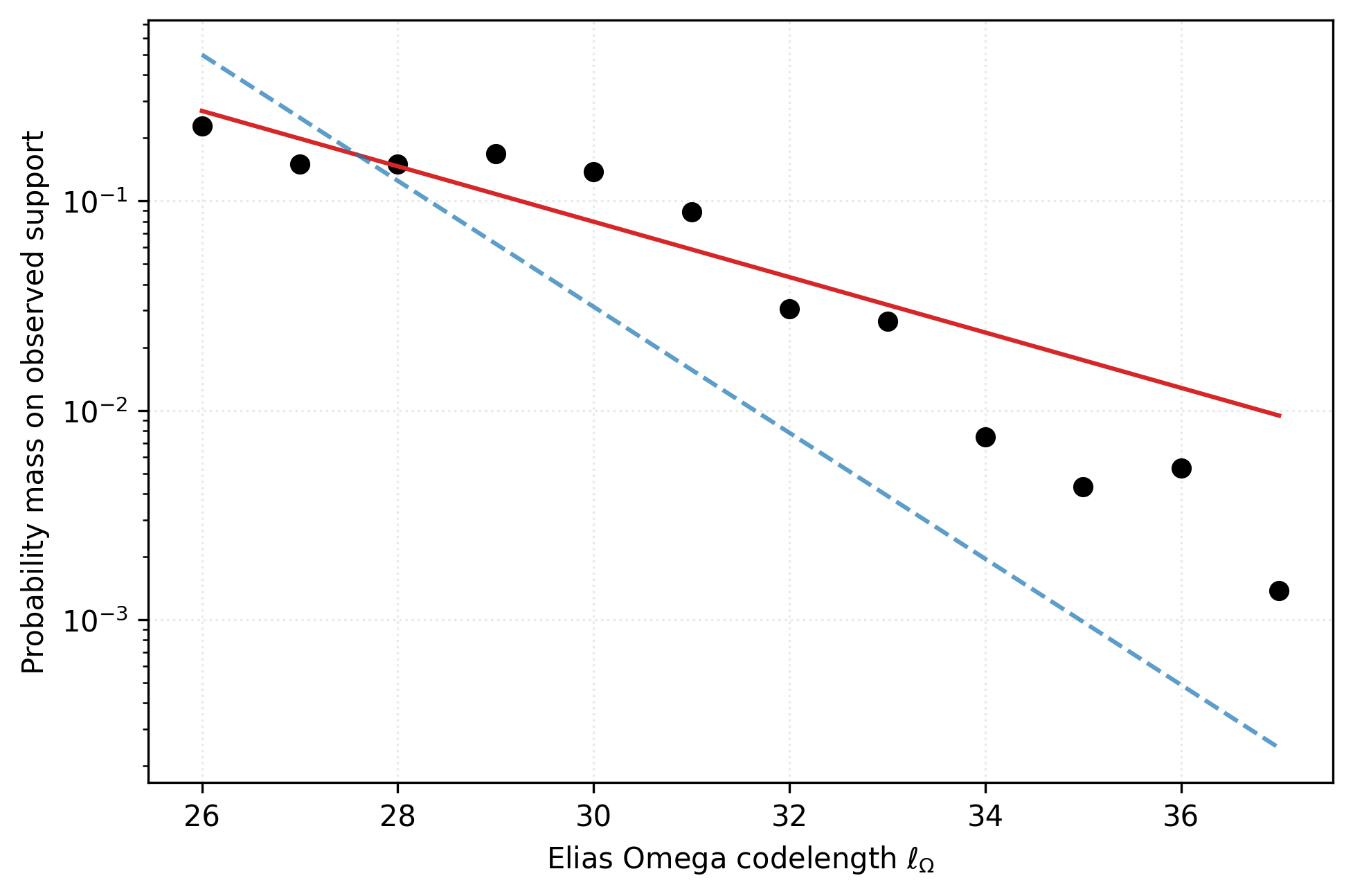}
\caption{PyPI latest-release file sizes for the top 750 projects. Black points are empirical codelength frequencies. Left: blue line is the pure Omega reference. Right: red line is the fitted scaled-Omega law and dashed blue line is the pure Omega reference.}
\label{fig:pypi-omega}
\end{figure}

\begin{figure}[htbp]
\centering
\includegraphics[width=0.47\textwidth]{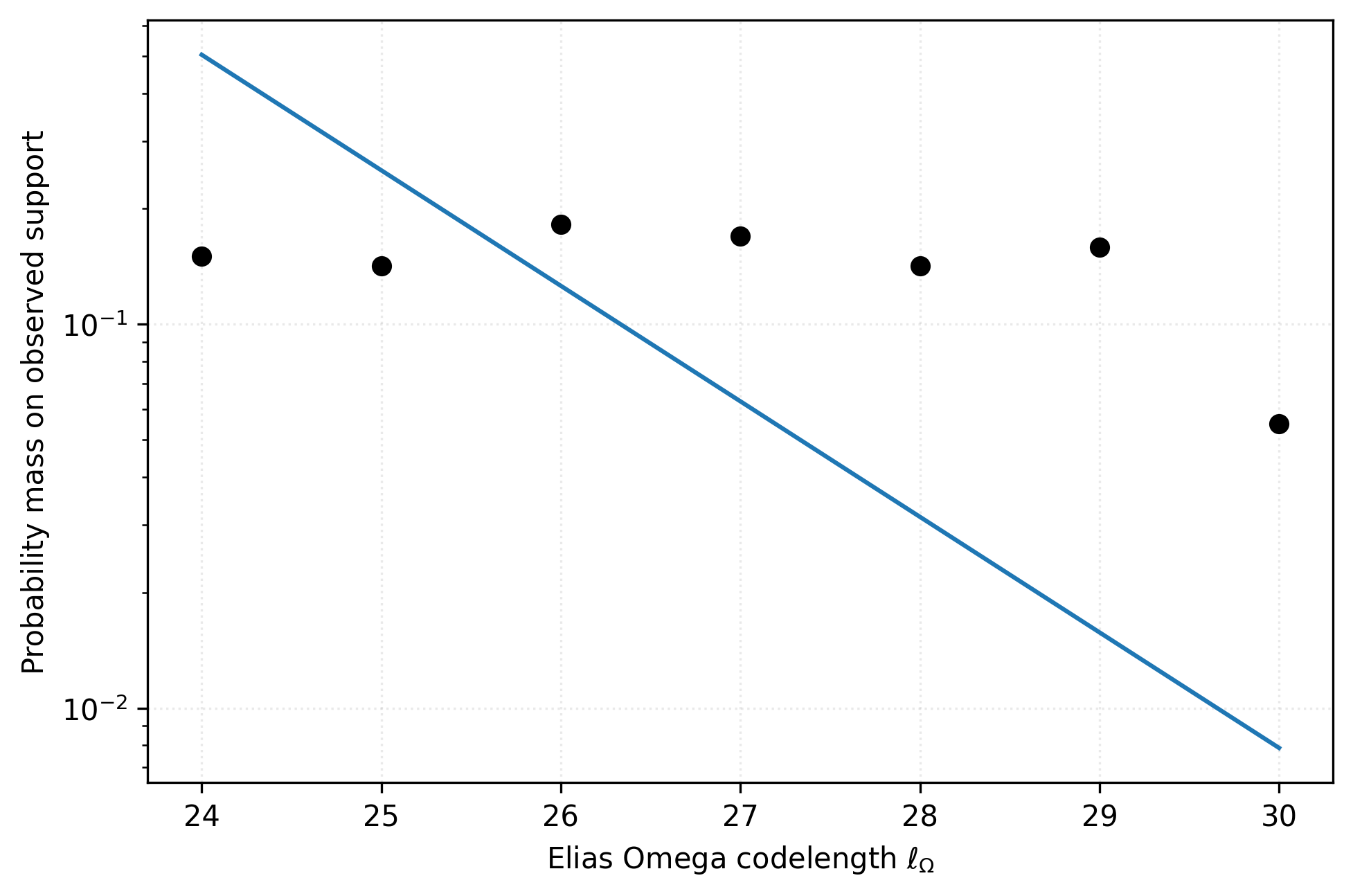}\hfill
\includegraphics[width=0.47\textwidth]{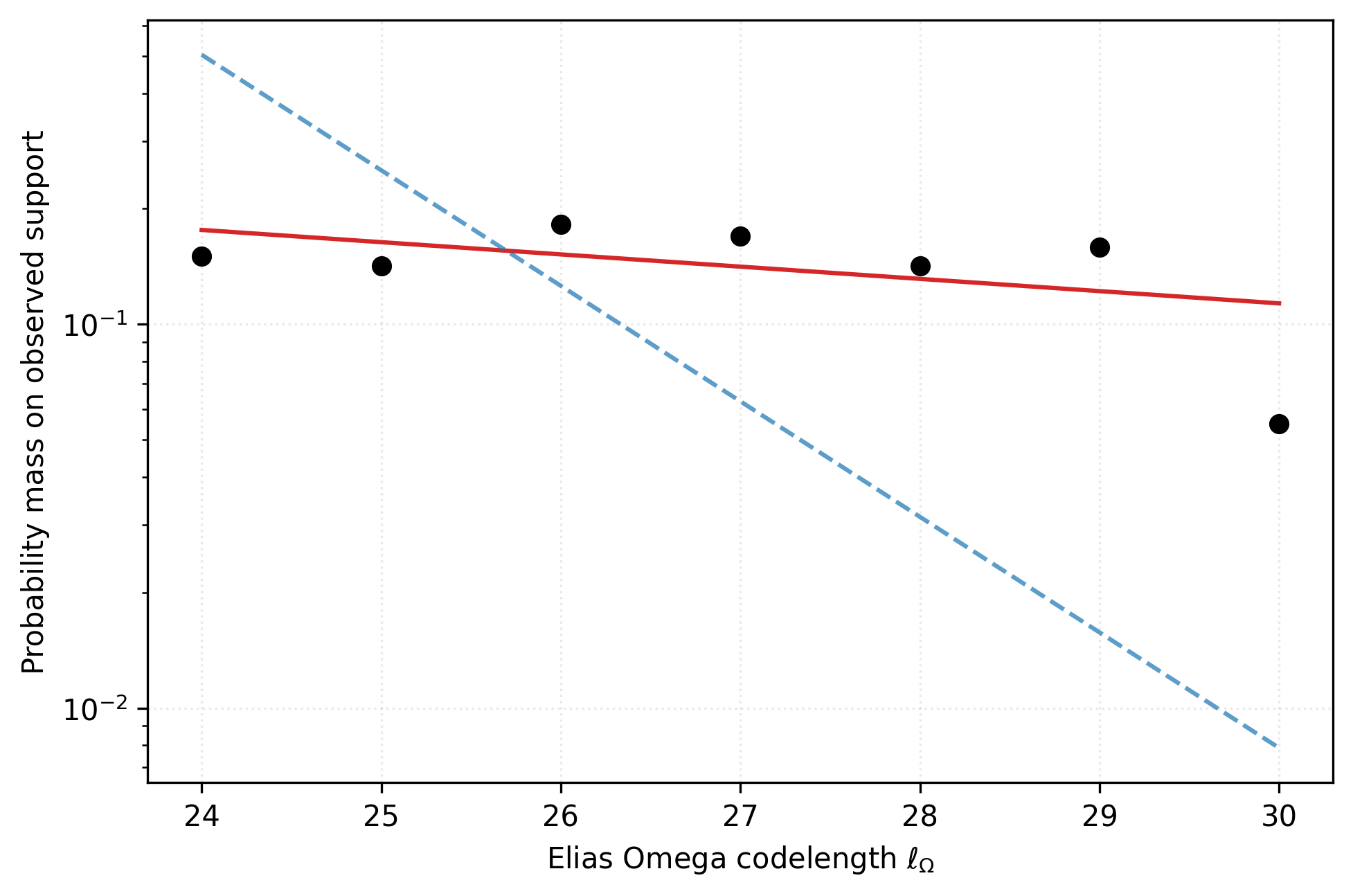}
\caption{CRAN source package archive sizes for a deterministic sample of 750 packages. Black points are empirical codelength frequencies. Left: blue line is the pure Omega reference. Right: red line is the fitted scaled-Omega law and dashed blue line is the pure Omega reference.}
\label{fig:cran-omega}
\end{figure}

\paragraph{Discussion.} In all three datasets the fitted slope is far below $\log 2\approx 0.693$. Since $q_a(\ell)$ decays more slowly when $a$ is smaller, these histograms are heavier-tailed than the pure Omega prior on the observed support. Equivalently, the effective scaled-Omega exponents $\widehat\beta=\widehat a/\log 2$ are approximately $0.285$ for Debian, $0.440$ for PyPI, and $0.106$ for CRAN. These values are not interpreted as universal constants of software production. They indicate that the observed software-size distributions lie on the high-variability side of the Omega energy scale. The empirical codelengths are proxies for integer scale and description length, not direct measurements of Kolmogorov complexity or psychology.

\section{Conclusion}
We introduced the Multiplicative Turing Ensemble (MTE), a Markov chain on positive integers driven by i.i.d.\ prime multipliers. The Maximum Entropy Principle applied to Elias' $\omega$ code length yields a natural prior on prime multipliers, though the pure $\omega$ prior $\pi_p \propto 2^{-\ell_\omega(p)}$ is a boundary case ($\lambda=1$) with infinite first moment. The \emph{scaled} $\omega$ prior $\pi_p \propto 2^{-\beta \ell_\omega(p)}$ with $\beta>1$ has finite moments and yields exponential tails for log-multipliers (modulo slow variation), which in turn generate asymptotically Pareto-type gap survival tails with exponent $\beta-1$.

Along MTE trajectories, the $\omega$ code length satisfies an almost-sure averaging law, though not ergodicity, since the chain is transient. Empirically, the fitted finite-support Omega slopes for Debian, PyPI, and CRAN are $0.1979$, $0.3047$, and $0.0734$, respectively, far below the pure-Omega slope $\log 2$. The corresponding KL divergences improve from $1.6066$ to $0.0998$ for Debian, from $0.4226$ to $0.0734$ for PyPI, and from $0.6598$ to $0.0330$ for CRAN. Within the Omega-family diagnostic, this indicates a heavier-than-pure-Omega tail regime, with effective exponents $\widehat\beta=\widehat a/\log 2<1$.

While Theorems~\ref{th:conditional} and~\ref{th:mixture-tail} require $\beta > 1$ for well-behaved Pareto gap asymptotics, the regime $\beta \leq 1$ is not pathological---it reflects systems with high variability and scale-free structure. For instance, MTEs with $\beta \approx 1$ exhibit Benford's law (logarithmic digit distributions) \cite{KR2025}, a phenomenon ubiquitous in natural datasets. The fitted $\beta < 1$ values suggest that real-world integer distributions encode greater diversity and long-tail phenomena than the pure algorithmic baseline would predict.

One possible interpretation is that the pure $\omega$ prior ($\beta=1$) represents a baseline computational model where integers are weighted solely by code length. Real systems---shaped by productive genius---exhibit heavier tails ($\beta < 1$), reflecting the presence of exceptional outliers and creative breakthroughs. The theoretical regime $\beta > 1$ ensures tractable asymptotics but may correspond to overly constrained distributions lacking the extreme contributions that characterize human-driven processes.

\section{Data availability}

All data and code used to produce this manuscript are available on GitHub~\cite{github-pareto}. 




\printbibliography

@article{Mandelbrot74,
  author  = {Mandelbrot, B.},
  title   = {{Intermittent turbulence in self-similar cascades: divergence of high moments and dimension of the carrier}},
  journal = {J. Fluid Mech.},
  volume  = {62},
  year    = {1974},
  pages   = {331--358}
}

@article{Kesten73,
  author  = {Kesten, H.},
  title   = {{Random difference equations and renewal theory for products of random matrices}},
  journal = {Acta Math.},
  volume  = {131},
  year    = {1973},
  pages   = {207--248}
}

@book{BGT1987,
  author    = {Bingham, N. H. and Goldie, C. M. and Teugels, J. L.},
  title     = {{Regular Variation}},
  publisher = {Cambridge Univ. Press},
  year      = {1987}
}

@article{Elias1975,
  author  = {Elias, Peter},
  title   = {{Universal codeword sets and representations of the integers}},
  journal = {IEEE Trans. Inf. Theory},
  volume  = {21},
  number  = {2},
  pages   = {194--203},
  month   = mar,
  year    = {1975},
  doi     = {10.1109/TIT.1975.1055349}
}

@book{li1997introduction,
  author    = {Li, Ming and Vit{\'a}nyi, Paul},
  title     = {{An Introduction to Kolmogorov Complexity and Its Applications}},
  publisher = {Springer},
  edition   = {2nd},
  year      = {1997}
}

@book{calude2002random,
  author    = {Calude, Cristian S.},
  title     = {{Information and Randomness: An Algorithmic Perspective}},
  publisher = {Springer},
  edition   = {2nd},
  year      = {2002}
}

@book{arora2009computational,
  author    = {Arora, Sanjeev and Barak, Boaz},
  title     = {{Computational Complexity: A Modern Approach}},
  publisher = {Cambridge Univ. Press},
  year      = {2009},
  isbn      = {9780521424264}
}

@misc{KR2025,
  author        = {Kolpakov, Alexander and Rocke, Aidan},
  title         = {{Benford's Law from Turing Ensembles and Integer Partitions}},
  year          = {2025},
  month         = jun,
  eprint        = {2502.16314},
  archivePrefix = {arXiv},
  primaryClass  = {cs.IT},
  doi           = {10.48550/arXiv.2502.16314}
}

@misc{KR-elias-2025,
  author        = {Kolpakov, Alexander and Rocke, Aidan},
  title         = {{Elias' Encoding from Lagrangians and Renormalization}},
  year          = {2025},
  month         = jun,
  eprint        = {2506.23447},
  archivePrefix = {arXiv},
  primaryClass  = {cs.IT}
}

@techreport{Kraft1949,
  author      = {Kraft, L. G.},
  title       = {{A Device for Quantizing, Grouping, and Coding Amplitude Modulated Pulses}},
  institution = {Massachusetts Institute of Technology},
  year        = {1949}
}

@article{McMillan1956,
  author  = {McMillan, Brockway},
  title   = {{Two inequalities implied by unique decipherability}},
  journal = {IRE Trans. Inf. Theory},
  volume  = {2},
  number  = {4},
  pages   = {115--116},
  year    = {1956},
  doi     = {10.1109/TIT.1956.1056816}
}

@misc{shen2015,
    author = {Shen, Alexander},
    title  = {{Around Kolmogorov complexity: basic notions and results}},
    year   = {2015},
    eprint = {1504.04955},
    archivePrefix = {arXiv},
    primaryClass = {cs.IT},
}

@article{github-pareto,
  author  = {Kolpakov, Alexander and Rocke, Aidan},
  title   = {{Auxiliary code for ``Multiplicative Turing Ensembles ...''}},
  journal = {GitHub},
  year    = {2025},
  url     = {https://github.com/sashakolpakov/mte-pareto}
}

\end{document}